\documentclass[11pt]{article}

\usepackage{csquotes}
\usepackage{graphics,graphicx,varwidth}
\usepackage[dvips]{epsfig}
\usepackage{amsmath,amssymb,amsfonts}
\usepackage{color}
\usepackage{bbm}
\usepackage{times}

\usepackage[lined]{algorithm2e}
\usepackage{microtype}
\usepackage{ifpdf}
\usepackage{paralist}

\usepackage{float}
\floatstyle{boxed}
\newfloat{pseudocode}{thb}{pseudo}
\floatname{pseudocode}{Algorithm}
\usepackage{multirow}
\usepackage{microtype,mathtools,mathrsfs}
\usepackage{amsthm}
\usepackage[english]{babel}
\usepackage{cite}
\usepackage{hyperref}
\usepackage[small,compact]{titlesec}
\usepackage[small,it]{caption}
\usepackage[margin=1in]{geometry}
\usepackage{subcaption}

\newcommand{\Cover}{\text{\tt Cover}}
\newcommand{\SPoA}{\text{\tt SPoA}}

\newcommand\RR{\mathbb{R}}

\newcommand{\Natural}{\mathbb{N}}

\newtheorem{thm}{Theorem}
\newtheorem{corollary}[thm]{Corollary}
\newtheorem{claim}[thm]{Claim}

\newtheorem{obs}[thm]{Observation}
\newcommand{\cP}{\cal{P}}
\newif\ifdraft
\drafttrue

\begin{document}
\title{Intense Competition can Drive Selfish Explorers to \\ Optimize  Coverage\footnote{This work has received funding from the European Research Council (ERC) under the European Union's Horizon 2020 research and innovation programme (grant  agreement No 648032).}}
\date{}
%The Tendency for Exclusiveness can Drive Selfish Individuals into a Strategy that Maximizes the Social Welfare 
%Optimal Cooperative Performances can Emerge\\ from a Tendency for Exclusiveness
%\\Optimal Cooperative Performances can Emerge\\ from Highly Competitive Conditions

\author{Simon Collet
\thanks{IRIF, CNRS and University Paris Diderot, Paris, France.  E-mail: {\tt Simon.Collet@irif.fr}.}
 \and
Amos Korman
\thanks{
IRIF, CNRS and University Paris Diderot, Paris, France.  E-mail: {\tt Amos.Korman@irif.fr}. }}

%\begin{titlepage}
\maketitle

\thispagestyle{empty}

\begin{abstract}%\parskip0.08cm
We consider a game-theoretic setting in which selfish individuals
compete over resources of varying quality.  The motivating example is
a group of animals that disperse over patches of food of different
abundances. In such scenarios, individuals are biased towards
selecting the higher quality patches, while, at the same time, aiming
to avoid costly collisions or overlaps. Our goal is to investigate the
impact of collision costs on the parallel coverage of resources by the
whole group.

Consider $M$ {\em sites}, where a site $x$ has {\em value} $f(x)$. We think of $f(x)$ as the reward associated with site $x$, 
and assume that if a single individual visits $x$ exclusively, it receives this exact reward. Typically, we assume that if $\ell>1$ individuals visit $x$ then each receives at most $f(x)/\ell$. In particular, when competition costs are high, each individual might receive an amount strictly less than $f(x)/\ell$, which could even be negative. Conversely, modeling cooperation at a site, we also consider cases where each one gets more than $f(x)/\ell$. There are $k$ identical 
{\em players} that compete over the rewards. They independently act in parallel, in a
one-shot scenario, each specifying a single site to visit, without
knowing which sites are explored by others.  The group performance is evaluated by 
 the expected {\em coverage}, defined as the sum of
$f(x)$ over all sites that are explored by at least one player. Since we assume that players cannot coordinate before choosing their site we focus on symmetric strategies. 

The main takeaway message of
this paper is that the optimal symmetric coverage is expected to emerge when collision costs are relatively high, so that the following  ``Judgment of Solomon'' type of rule holds:
If a single player explores a site $x$ then it gains its full reward
$f(x)$, but if several players explore it, then neither one receives
any reward. Under this policy, it turns out that there exists a unique
symmetric Nash Equilibrium strategy, which is, in fact, evolutionary
stable. Moreover, this strategy yields the best possible coverage
among all symmetric strategies. Viewing the coverage measure as the social welfare, this policy thus enjoys a {\em (Symmetric) Price
  of Anarchy} of precisely 1, whereas, in fact, any other
congestion policy has a price strictly greater than 1.  

Our model falls within the scope of mechanism design, and more
precisely in the area of incentivizing exploration. It finds
relevance in evolutionary ecology, and further connects to studies on
Bayesian parallel search algorithms.

\end{abstract}
%\bigskip
%\vspace{0.9cm}
%\end{titlepage}
\pagenumbering{arabic}

%\parskip0.1cm
%%%%

\section{Introduction}\label{sec:intro}
Studying the way humans and other animals distribute themselves in their environment is a cornerstone of ecology and the social sciences \cite{Amoebae,Social-foraging,hills, broom2013game}. 
In many of these contexts, dispersal is governed by two contradicting
forces. On the one hand, individuals are biased towards selecting the  higher quality patches, and, on the other hand, they aim to avoid costly collisions or overlaps, which can significantly deteriorate the value of a patch \cite{IFD-review,IFD-review2,FL69}.  
In nature, collision costs can be caused by various factors, including aggressive behavior,
or merely due to sharing the patch equally between colliding individuals (a.k.a., scramble competition \cite{Nicholson}). 
The purpose of this paper is to investigate the impact of collision costs on the overall coverage of resources by the whole group.

Understanding the impact of collisions on the overall parallel coverage  can also have implications to the way centralized entities incentive individuals in human organizations. For example, research foundations are often interested in promoting specific research topics, and offer  
 grants to attract researchers to such topics. The question of how to properly define a grant policy that is expected to yield a desirable distribution of researchers can have significant implication on the scientific community \cite{kleinberg2011mechanisms}.

This paper proposes a framework to study the impact of collision on the group performances through an {\em algorithmic mechanism design} approach. We focus on a relatively simple scenario, which already provides interesting, and perhaps surprising, insights. 
See Section \ref{sec:conclusions} for suggested generalizations left for future work. 

\subsection{The Dispersal Game} \label{sec:dispersal}
Think about the following imaginary scenario. A group of individuals are searching for resources in some space. Eventually, they gather all resources found for the benefit of the public, and hence the group performances is defined as the sum of the resources values. In order to avoid the free-rider problem, the group incentivizes individuals by granting those who find more resource with higher social status. When an individual exclusively finds a resource of value $f$, a simple mechanism is to let its social status  be proportional to $f$. However, how should the group define the social status of each of $\ell$ individuals in case they all found a resource simultaneously?

Formally, we have  $M$ {\em sites} indexed by $x\in [1,M]$. Each site $x$ has some {\em
  importance value} $f(x)$, and we assume without loss of generality
that sites are ordered such that lower index sites have higher values,
i.e., $f(x)\geq f(x+1)$ for each $x<M$.

We have $k$ players that act a one-shot scenario, each specifying a cite $x$ to visit. Crucially, they cannot coordinate, and each player must commit to a site $x$, without knowing which sites are selected by other
players \cite{STOC}.  Moreover, we assume that players cannot change their mind after committing. 
Formally, each player specifies an index $x$ according to some law, called {\em strategy}. This is simply a probability distribution, indicating, for each $x$, the probability $p(x)$ to explore  $x$.  
A {\em strategy profile} is a vector of $k$ strategies, one for each player. It is called {\em symmetric} when all  players play the same strategy. For short, we shall use the term {\em symmetric strategy} to refer to the strategy being played in the case of a symmetric strategy profile. 

The group performance is  evaluated by the {\em weighted coverage}, defined as the sum of the values $f(x)$ over  visited sites $x$. With the lack of coordination between users, the most the group can hope for is the best achievable coverage among all symmetric strategies, that is, when assuming that all players play the same strategy.
 Formally, the {\em coverage} of a strategy $p$ is defined as:
\begin{equation}\label{eq:collaborative-max}
{\Cover}(p)= \sum_{x=1}^M f(x)(1-(1-p(x))^k)
\end{equation}
Let $p^{\star}$ be a  strategy that maximizes ${\Cover}(p)$. Using a compactness argument, it is not difficult to show that such a strategy exists. 
The following observation implies that the optimal symmetric strategy $p^\star$ has coverage which is, up to a small constant factor, as good as the best achievable coverage in the full coordination case, when the $k$ users can be assigned to precisely cover the $k$ highest value sites. 
\begin{obs}\label{obs-pstar}
$\Cover(p^\star)>(1-\frac{1}{e})\sum_{x=1}^k f(x)$.
\end{obs}
\begin{proof}
Let $\hat{p}$ be the uniform distribution in $[k]$, that is, we have $\hat{p}(x)=1/k$ for each $x\leq k$. Then,  
$\Cover(p^\star)\geq {\Cover}(\hat{p})= \sum_{x=1}^k f(x)\left(1-\left(1-\frac{1}{k}\right)^k\right)>(1-\frac{1}{e})\sum_{x=1}^k f(x)$. 
\end{proof}

A {\em reward policy} is a function $I(x,\ell)$ specifying the {\em payoff} (or reward) that a player exploring $x$ receives, given that $\ell-1$ other players are also exploring $x$.  
We shall be particularly interested in {\em congestion} reward policies, which  can be written as: \[I(x,\ell)=f(x)\cdot C(\ell)\]
where  $C(\cdot)$ is some  {\em congestion} function, for which  $C(1)=1$, and $C$ is non-increasing. Note that the specification of the policy does not depend on the total number of players $k$, and is influenced only by the number of players that visit $x$. 

The assumption that $C(1)=1$ implies that a player visiting $x$
  exclusively will gain exactly\footnote{Alternatively, we could have
    defined $C(1)$ to be any constant that doesn't depend on $x$. This
    would mean that the reward a single player gets when visiting $x$
    exclusively is proportional to the value $f(x)$. (This assumption
    contrasts with the setting of \cite{kleinberg2011mechanisms}, see
    more details in Section \ref{sec:related}.) Note that in this
    sense, $f(x)$ plays a double role - it has a meaning from the
    perspective of the group and a meaning from the perspective of an
    individual.} $f(x)$. A natural example of a congestion policy is
  the {\em sharing policy} in which $C_{share}(\ell)=1/\ell$, and
  hence if $\ell$ players select site $x$ in parallel then each
  receives a payoff of $f(x)/\ell$. This policy has received a lot of
  attention in both the ecology literature and algorithmic game theory
  community due to its simplicity
  \cite{IFD-review,IFD-review2,kleinberg2011mechanisms}. Cases where
  $C(\ell)>1/\ell$ can model cooperation, in the sense that the
  presence of several players at a site would result in each having
  more than their relative share. On the other extreme, note that we
  do not restrict $C$ to be a positive function. Taking $C(\ell)$ to
  be negative implies that the presence of $\ell$ players at $x$
  damages each of them quite significantly. In the context of animals,
  this could represent aggressive behavior that increases with the
  amount of potential benefit $f(x)$ and could result in severe
  injuries.

 Let $X_i$ be the site specified by player $i$. Given a policy $I$, the goal of each player $i$ is to
maximize the {\em expected payoff} it receives, defined as 
\[\mathbb{E} \big[I(X_i, \sum_{j} \mathbbm{1}_{X_j=X_i})\big]\]

\subsection{Evaluating Policies}

The game we define is symmetric and possesses  both mixed symmetric
equilibria and pure non-symmetric equilibria. Most previous works dealing with such situations prefer to focus on pure equilibria (rather than on the mixed equilibria),
because of their conceptual simplicity. The classical example of this
situation is the well-know game ``battle of the sexes'' \cite{bos}. However we argue otherwise:
the number of pure equilibria grows exponentially with the number of
players, and choosing an equilibrium among those requires coordination
between the players. In some contexts, because coordination is not
possible or communication is too costly, it does not seem reasonable
to expect that a pure equilibrium arises in practice, and instead, it
is more natural to expect the emergence of a symmetric
equilibrium. Since we focus on large distributed systems, where
players do not always possess individual identities and often have
identical behaviors, we will restrict our attention to symmetric
equilibria, and disregard pure equilibria\footnote{At first glance, our focus on symmetric strategies may seem non-compatible with the fact that conspecific animals often appear not to have identical strategies. However, it is important to note that a symmetric strategy profile in the context of ESS (see Section~\ref{sec:ESS}) does not imply that every player follows the same strategy! Indeed, a symmetric strategy profile can arise also if the population contains pure strategies in proportions as specified by the mixed strategy. In this case, since the framework of ESS assumes that at each time $k$ players are selected at random from the population, each of the selected random players ends up having the same mixed strategy. For this reason, restricting attention to
symmetric equilibria is in fact very common in the ecology literature \cite{IFD-review,IFD-review2,FL69}.}.

Viewing $p^{\star}$ as the symmetric strategy that maximizes the social welfare, we adopt an interpretation of the concept of Price of Anarchy to evaluate a policy $I(x,\ell)$, by measuring the coverage of the worst symmetric Nash Equilibrium under $I(x,\ell)$ in comparison with the coverage of $p^{\star}$. 
Specifically,  let $C(\ell)$ be a congestion function, and $f(x)$ be an importance value
 function. Define 
\[ \SPoA(C, f):=\sup_{p\in{\cal{P}}(C,f)}\left\{\frac{\Cover(p^{\star})}{\Cover(p)} \right\},\] where ${\cP}(C,f)$ is the set of symmetric Nash Equilibrium under $I(x,\ell)=f(x) \cdot C(\ell)$. 
 The  {\em Symmetric Price of Anarchy $(\SPoA)$}   of the congestion function $C(\ell)$ is  defined as the sup of $\SPoA(C, f)$ over all functions $f:[M]\rightarrow \RR^+$, where we take the freedom to choose their domain $M$ as well, i.e.,
\[
  \SPoA(C):=\sup_{\stackrel{M \in \Natural}{f:[M]\rightarrow \RR^+}} \left\{\SPoA(C,f) \right\}
\]
By definition,  $\SPoA(C)\geq 1$ for any congestion function $C$. On the other hand, 
there exist congestion functions  with much higher $\SPoA$. For example, taking $C(\ell)\equiv 1$ yields SPoA of roughly $k$. Indeed, under this function, and assuming a strictly decreasing $f$, a strategy at equilibrium would explore the first site with probability 1, and for distributions $f$ that decrease very slowly (and thus are close to uniform), this would yield a gap $\Omega(k)$, with respect to, e.g.,  the uniform strategy. Note, however, that this policy is unlikely to occur in nature, as it implies that a value at a site is fully obtained by all visitors.

The PoA of
  the sharing policy $C_{share}$ was studied by Kleinberg and Oren in a somewhat similar
  model \cite{kleinberg2011mechanisms} (see also Section \ref{sec:related}). Thanks to a more general
  result from Vetta \cite{vetta2002nash}, they showed that the SPoA of
  the sharing policy, defined as the ratio between the best achievable
  coverage and the worst coverage among Nash equilibria, is at most
  2. Since the worst symmetric equilibrium coverage is at least as
  good as the worst equilibrium coverage, and the best achievable
  coverage is at least as good as the best achievable coverage with a
  symmetric strategy profile, this directly implies that our notion of
  SPoA also satisfies $\SPoA(C_{share}) \leq 2$.

\subsection{Ideal Free Distribution (IFD)}
The setting of animals competing over patches of resources, often referred to as {\em playing the field}, has been extensively studied in the ecology discipline through the theory of {\em Ideal Free Distribution (IFD)}, see reviews in \cite{IFD-review,IFD-review2}. Introduced by Fretwell and Lucas \cite{FL69}, the standard model assumes that each individual follows the same distribution $p(\cdot)$ (i.e., a symmetric strategy) and hence the fraction of the population that is expected to occupy a site $x$ is $p(x)$. 

An IFD is a probability distribution $p(\cdot)$  in which every site which is explored with positive
probability yields the same payoff, and every other site yields a lower payoff if explored. 
More precisely, as before, let $I(x,\ell)$ denote the payoff received by an individual selecting a site $x$ together with $\ell-1$ other individuals. 
Fix a player and let $P(x,\ell) = \binom{k-1}{\ell}
p(x)^{\ell}(1-p(x))^{k-\ell-1}$ be the probability that among 
$k-1$  players, $\ell$ of them selected site $x$. The 
  $value_p$ of a site $x$ corresponds to the expected gain for
exploring site $x$ and is defined as:
\begin{equation}\label{p-value}
  \nu_ p(x):=  \sum_{\ell=1}^{k}  I(x,\ell) \cdot P(x,\ell-1)
\end{equation}
By definition, the expected payoff that a player receives  is 
$\sum_x p(x)\cdot \nu_ p(x)$.
The IFD is a distribution $p$, that for some integer $W$, satisfies
the following {\em IFD conditions.} 
\begin{itemize}
\item $p(x)>0 \iff x\leq W$
\item%\label{eq-val}
$\nu_ {p}(1)=\nu_ {p}(2)=\ldots= \nu_ {p}(W)=\nu_ {p}$, ~and for all $W+1\leq x$, we have $\nu_ {p}>\nu_ {p}(x)$
%\item
%$\forall x\geq W+1$ we have $\nu_ {p}(W)>\nu_ {p}(x)$
\end{itemize}

\noindent The following observation is well-known, see e.g., \cite{FL69, CK10}. For completeness, a sketch of the proof is provided in Appendix \ref{app:obs}. 
\begin{obs}\label{obs}
Assume that $I(x,\cdot)$ is non-increasing for any $x$ (note that this is indeed the case for congestion policies). There exists a unique strategy satisfying the IFD conditions. Moreover,
this IFD is  the only symmetric Nash Equilibrium of the game.
\end{obs}

\subsection{Evolutionary Stable Strategy (ESS)}\label{sec:ESS}
In addition to the notion of Nash Equilibrium, we shall adopt the
stronger concept of an {\em Evolutionary Stable Strategy (ESS)} \cite{nature,BCV97, BM04, GT10,Linda, broom2013game}.
This concept has
become a cornerstone of evolutionary game theory, and has been 
 extensively used to   
study the evolution of animal and human behavior \cite{Linda, broom2013game}. 
Informally, an ESS  is a strategy $\sigma$ such that, in a population composed of a
majority of $\sigma$-strategists and a minority of ``mutants'' playing
strategy $\pi$, the $\sigma$-strategists have a higher payoff,
provided that the proportion of mutants is not too high.

The standard view of ESS is to consider an infinite population
  whose individuals play repeatedly against each other in pairs,
  chosen uniformly at random. In our setting, individuals play in
  groups of $k$, and hence a generalization is required.
One possible generalization is to consider a fixed population of $k$ individuals that  play the game between themselves. Unfortunately, generalizing ESS to finite populations is quite challenging, and the known generalizations are often too restrictive to be applicable. Instead, we chose the generalization in which the population remains infinite and players are randomly matched in $k$-tuples, for any fixed $k\geq 2$, see  e.g., the book \cite{broom2013game}. Modeling biological settings, this generalization may correspond, for example, to scenarios in which a large colony of bats breaks on a daily basis into smaller groups, each of which forages in a different field of patched resources, see \cite{follow,Dechmann2721}.

Formally, let $\mathcal{E}(\rho;\sigma^{\ell},\pi^{k-\ell-1})$ denote
the payoff received by an individual playing strategy $\rho$ against
$\ell$ players playing strategy $\sigma$ and $(k-\ell-1)$ players
playing strategy $\pi$. When $\sigma = \pi$, we use the abridged
notation $\mathcal{E}(\rho;\sigma^{k-1})$. Consider a population in
which a proportion $(1-\epsilon)$ of it are playing strategy $\sigma$
and a proportion of $\epsilon$ are playing $\pi$. The average payoff
of a player playing $\rho$ matched against $k-1$ opponents drawn
uniformly at random from this population is denoted by
$\mathcal{U}[\rho;(1-\epsilon)\sigma+\epsilon \pi]$, and is equal to
\begin{equation}
  \mathcal{U}[\rho;(1-\epsilon)\sigma+\epsilon \pi] =
  \sum_{\ell = 0}^{k-1} \binom{k-1}{\ell}(1-\epsilon)^\ell
  \epsilon^{k-\ell-1} \mathcal{E}(\rho;\sigma^{\ell},\pi^{k-\ell-1})
\end{equation}
A strategy $\sigma$ is an ESS if for every other strategy $\pi$, there
exist $\epsilon_{\pi} > 0$ such that for every
$\epsilon \leq \epsilon_{\pi}$, 
%\begin{equation}
$ \mathcal{U}[\sigma, (1-\epsilon)\sigma+\epsilon \pi] > \mathcal{U}[\pi,
    (1-\epsilon)\sigma+\epsilon \pi]$, 
%\end{equation}
i.e.,  $\sigma$ cannot be ``invaded'' by $\pi$, as long as the proportion of
$\pi$-strategists remains below $\epsilon_{\pi}$.

We also recall the following {\em characterization of ESS} \cite{broom2013game}: a strategy $\sigma$ is an ESS if for every
other strategy $\pi$, there is an integer  $0\leq m_{\pi} \leq k-1$, such that
both conditions below hold:
\begin{itemize}
\item $\mathcal{E}(\sigma;\sigma^{k-m_\pi-1}, \pi^{m_\pi}) > \mathcal{E}(\pi;\sigma^{k-m_\pi-1}, \pi^{m_\pi})$
\item $\forall \ell,~ 0\leq \ell < m_\pi, \mbox{~we have}~\mathcal{E}(\sigma;\sigma^{k-\ell-1}, \pi^\ell) = \mathcal{E}(\pi;\sigma^{k-\ell-1}, \pi^\ell)$
\end{itemize}
It should be clear from this definition that every ESS is a symmetric Nash
equilibrium of the $k$-player game. Indeed, these conditions ensure
that $\mathcal{E}(\sigma; \sigma^{k-1}) \geq \mathcal{E}(\pi;
\sigma^{k-1})$ for every $\pi \neq \sigma$.
Further details about the ESS notion appear in Section \ref{sec:conclusions}.

\subsection{Our Results}
We study the interplay between the congestions and the expected coverage of strategies at equilibrium.   
Our main takeaway message is that optimal coverage is expected to emerge by a congestion policy which is 
defined by a ``Judgment of Solomon'' type of rule: If a single player attempts to exploit a  site then it gains the full reward $f(x)$, but when several players attempt to exploit it, they all receive nothing. 
Formally, define the {\em exclusive congestion function} as follows: 
\begin{align*}
C_{exc}(\ell) =
  \begin{cases}
    1 &\quad \text{if } \ell=1\\
    0 &\quad \text{otherwise}
  \end{cases} 
  \end{align*}
and define the
{\em exclusive reward policy} as \[I_{exc}(\ell,x) =f(x)\cdot C_{exc}(\ell)\]
We identify $I_{exc}$ as a critical reward policy, in the sense that 
its IFD yields optimal coverage, while any other congestion policy yields
strictly worse coverage.

More formally, let $\sigma^{\star}$ be the IFD for the exclusive reward policy.
We first claim that under this policy, $\sigma^{\star}$ is not only a Nash Equilibrium but also an ESS.
 
  \begin{thm}\label{thm-ESS}
 $\sigma^{\star}$ is an ESS under $I_{exc}$.  
\end{thm} 

\noindent Next, we claim that $\sigma^{\star}$ yields the best possible coverage, among all symmetric strategies.
  \begin{thm}
  \label{thm-optimum}
 For any strategy $\sigma$, ${\Cover}(\sigma^{\star})\geq {\Cover}(\sigma)$, with equality if and only if $\sigma=\sigma^{\star}$.
 \end{thm}
\noindent Since the IFD is the only symmetric Nash Equilibrium under
$I_{exc}$, the price of anarchy equals 1.
 \begin{corollary}
  $\SPoA(C_{exc})=1$.
 \end{corollary}

\noindent 
The criticality of $I_{exc}$ follows by the fact that among all congestion policies, it is the only one whose IFD yields optimal coverage for all value functions $f$. 
 
 \begin{thm}\label{thm-critical}
   For any congestion function $C\neq C_{exc}$, we have $\SPoA(C)>1$.
\end{thm} 
\noindent 
We find the fact that $\SPoA(C_{exc})=1$ rather surprising. Indeed, although it appears intuitive that increasing the competition will result in better coverage, the exclusive policy $C_{exc}$ is, in fact, not the one with highest levels of competition. Indeed, one could define a congestion policy where in case of conflict, each of the colliding players receives a negative payoff (i.e., $C(\ell)$ being negative, see Section \ref{sec:dispersal}). This means that the competition level could significantly exceed the level of $C_{exc}$. The fact that the best coverage occurs exactly at that level is thus quite surprising. A second important factor to note is that even if one identifies the best policy (i.e., the one for which the $\SPoA$ is smallest), it is a priori unclear that the resulted coverage would actually be optimal among all symmetric strategies (including the non-competing ones), i.e., that $\SPoA = 1$.

An informal discussion regarding the implications of our results in the context of animal behavior is presented in Section \ref{sec:discussion}.

See Figure \ref{fig:coverage}  for an illustration of
the coverage as a function of the competition extent, in specific
instances of 2 players and 2 sites.

\begin{figure}[h]
    \centering
    \begin{subfigure}[b]{0.4\textwidth}
      % GNUPLOT: LaTeX picture with Postscript
      \begingroup
      \makeatletter
      \providecommand\color[2][]{%
        \GenericError{(gnuplot) \space\space\space\@spaces}{%
          Package color not loaded in conjunction with
          terminal option `colourtext'%
        }{See the gnuplot documentation for explanation.%
        }{Either use 'blacktext' in gnuplot or load the package
          color.sty in LaTeX.}%
        \renewcommand\color[2][]{}%
      }%
      \providecommand\includegraphics[2][]{%
        \GenericError{(gnuplot) \space\space\space\@spaces}{%
          Package graphicx or graphics not loaded%
        }{See the gnuplot documentation for explanation.%
        }{The gnuplot epslatex terminal needs graphicx.sty or graphics.sty.}%
        \renewcommand\includegraphics[2][]{}%
      }%
      \providecommand\rotatebox[2]{#2}%
      \@ifundefined{ifGPcolor}{%
        \newif\ifGPcolor
        \GPcolortrue
      }{}%
      \@ifundefined{ifGPblacktext}{%
        \newif\ifGPblacktext
        \GPblacktexttrue
      }{}%
      % define a \g@addto@macro without @ in the name:
      \let\gplgaddtomacro\g@addto@macro
      % define empty templates for all commands taking text:
      \gdef\gplbacktext{}%
      \gdef\gplfronttext{}%
      \makeatother
      \ifGPblacktext
      % no textcolor at all
      \def\colorrgb#1{}%
      \def\colorgray#1{}%
      \else
      % gray or color?
      \ifGPcolor
      \def\colorrgb#1{\color[rgb]{#1}}%
      \def\colorgray#1{\color[gray]{#1}}%
      \expandafter\def\csname LTw\endcsname{\color{white}}%
      \expandafter\def\csname LTb\endcsname{\color{black}}%
      \expandafter\def\csname LTa\endcsname{\color{black}}%
      \expandafter\def\csname LT0\endcsname{\color[rgb]{1,0,0}}%
      \expandafter\def\csname LT1\endcsname{\color[rgb]{0,1,0}}%
      \expandafter\def\csname LT2\endcsname{\color[rgb]{0,0,1}}%
      \expandafter\def\csname LT3\endcsname{\color[rgb]{1,0,1}}%
      \expandafter\def\csname LT4\endcsname{\color[rgb]{0,1,1}}%
      \expandafter\def\csname LT5\endcsname{\color[rgb]{1,1,0}}%
      \expandafter\def\csname LT6\endcsname{\color[rgb]{0,0,0}}%
      \expandafter\def\csname LT7\endcsname{\color[rgb]{1,0.3,0}}%
      \expandafter\def\csname LT8\endcsname{\color[rgb]{0.5,0.5,0.5}}%
      \else
      % gray
      \def\colorrgb#1{\color{black}}%
      \def\colorgray#1{\color[gray]{#1}}%
      \expandafter\def\csname LTw\endcsname{\color{white}}%
      \expandafter\def\csname LTb\endcsname{\color{black}}%
      \expandafter\def\csname LTa\endcsname{\color{black}}%
      \expandafter\def\csname LT0\endcsname{\color{black}}%
      \expandafter\def\csname LT1\endcsname{\color{black}}%
      \expandafter\def\csname LT2\endcsname{\color{black}}%
      \expandafter\def\csname LT3\endcsname{\color{black}}%
      \expandafter\def\csname LT4\endcsname{\color{black}}%
      \expandafter\def\csname LT5\endcsname{\color{black}}%
      \expandafter\def\csname LT6\endcsname{\color{black}}%
      \expandafter\def\csname LT7\endcsname{\color{black}}%
      \expandafter\def\csname LT8\endcsname{\color{black}}%
      \fi
      \fi
      \setlength{\unitlength}{0.0500bp}%
      \ifx\gptboxheight\undefined%
      \newlength{\gptboxheight}%
      \newlength{\gptboxwidth}%
      \newsavebox{\gptboxtext}%
      \fi%
      \setlength{\fboxrule}{0.5pt}%
      \setlength{\fboxsep}{1pt}%
      \begin{picture}(4320.00,3600.00)%
        \gplgaddtomacro\gplbacktext{%
          \csname LTb\endcsname%
          \put(773,596){\makebox(0,0)[r]{\strut{}$0.9$}}%
          \put(773,999){\makebox(0,0)[r]{\strut{}$0.95$}}%
          \put(773,1401){\makebox(0,0)[r]{\strut{}$1$}}%
          \put(773,1804){\makebox(0,0)[r]{\strut{}$1.05$}}%
          \put(773,2206){\makebox(0,0)[r]{\strut{}$1.1$}}%
          \put(773,2609){\makebox(0,0)[r]{\strut{}$1.15$}}%
          \put(773,3011){\makebox(0,0)[r]{\strut{}$1.2$}}%
          \put(1207,376){\makebox(0,0){\strut{}$-0.4$}}%
          \put(1810,376){\makebox(0,0){\strut{}$-0.2$}}%
          \put(2414,376){\makebox(0,0){\strut{}$0$}}%
          \put(3018,376){\makebox(0,0){\strut{}$0.2$}}%
          \put(3621,376){\makebox(0,0){\strut{}$0.4$}}%
          \put(2073,72){\makebox(0,0)[l]{\strut{}exclusive}}%
          \put(3671,72){\makebox(0,0)[l]{\strut{}sharing}}%
        }%
        \gplgaddtomacro\gplfronttext{%
          \csname LTb\endcsname%
          \put(176,831){\rotatebox{-270}{\makebox(0,0){\strut{}Coverage}}}%
          \put(988,154){\makebox(0,0){\strut{}$c$}}%
          \put(2414,3269){\makebox(0,0){\strut{}$f(x_1)=1, f(x_2)=0.3$}}%
          \csname LTb\endcsname%
          \put(2936,2838){\makebox(0,0)[r]{\strut{}ESS}}%
          \csname LTb\endcsname%
          \put(2936,2618){\makebox(0,0)[r]{\strut{}Welfare Optimum}}%
          \csname LTb\endcsname%
          \put(2936,2398){\makebox(0,0)[r]{\strut{}Optimum Coverage}}%
        }%
        \gplbacktext
        \put(0,0){\includegraphics{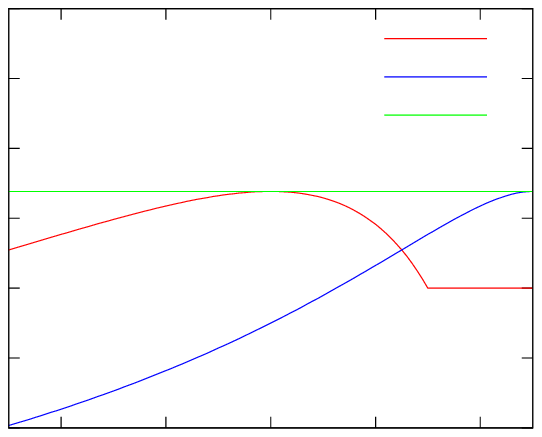}}%
        \gplfronttext
      \end{picture}%
      \endgroup
    \end{subfigure}
    ~
    \begin{subfigure}[b]{0.4\textwidth}
      % GNUPLOT: LaTeX picture with Postscript
      \begingroup
      \makeatletter
      \providecommand\color[2][]{%
        \GenericError{(gnuplot) \space\space\space\@spaces}{%
          Package color not loaded in conjunction with
          terminal option `colourtext'%
        }{See the gnuplot documentation for explanation.%
        }{Either use 'blacktext' in gnuplot or load the package
          color.sty in LaTeX.}%
        \renewcommand\color[2][]{}%
      }%
      \providecommand\includegraphics[2][]{%
        \GenericError{(gnuplot) \space\space\space\@spaces}{%
          Package graphicx or graphics not loaded%
        }{See the gnuplot documentation for explanation.%
        }{The gnuplot epslatex terminal needs graphicx.sty or graphics.sty.}%
        \renewcommand\includegraphics[2][]{}%
      }%
      \providecommand\rotatebox[2]{#2}%
      \@ifundefined{ifGPcolor}{%
        \newif\ifGPcolor
        \GPcolortrue
      }{}%
      \@ifundefined{ifGPblacktext}{%
        \newif\ifGPblacktext
        \GPblacktexttrue
      }{}%
      % define a \g@addto@macro without @ in the name:
      \let\gplgaddtomacro\g@addto@macro
      % define empty templates for all commands taking text:
      \gdef\gplbacktext{}%
      \gdef\gplfronttext{}%
      \makeatother
      \ifGPblacktext
      % no textcolor at all
      \def\colorrgb#1{}%
      \def\colorgray#1{}%
      \else
      % gray or color?
      \ifGPcolor
      \def\colorrgb#1{\color[rgb]{#1}}%
      \def\colorgray#1{\color[gray]{#1}}%
      \expandafter\def\csname LTw\endcsname{\color{white}}%
      \expandafter\def\csname LTb\endcsname{\color{black}}%
      \expandafter\def\csname LTa\endcsname{\color{black}}%
      \expandafter\def\csname LT0\endcsname{\color[rgb]{1,0,0}}%
      \expandafter\def\csname LT1\endcsname{\color[rgb]{0,1,0}}%
      \expandafter\def\csname LT2\endcsname{\color[rgb]{0,0,1}}%
      \expandafter\def\csname LT3\endcsname{\color[rgb]{1,0,1}}%
      \expandafter\def\csname LT4\endcsname{\color[rgb]{0,1,1}}%
      \expandafter\def\csname LT5\endcsname{\color[rgb]{1,1,0}}%
      \expandafter\def\csname LT6\endcsname{\color[rgb]{0,0,0}}%
      \expandafter\def\csname LT7\endcsname{\color[rgb]{1,0.3,0}}%
      \expandafter\def\csname LT8\endcsname{\color[rgb]{0.5,0.5,0.5}}%
      \else
      % gray
      \def\colorrgb#1{\color{black}}%
      \def\colorgray#1{\color[gray]{#1}}%
      \expandafter\def\csname LTw\endcsname{\color{white}}%
      \expandafter\def\csname LTb\endcsname{\color{black}}%
      \expandafter\def\csname LTa\endcsname{\color{black}}%
      \expandafter\def\csname LT0\endcsname{\color{black}}%
      \expandafter\def\csname LT1\endcsname{\color{black}}%
      \expandafter\def\csname LT2\endcsname{\color{black}}%
      \expandafter\def\csname LT3\endcsname{\color{black}}%
      \expandafter\def\csname LT4\endcsname{\color{black}}%
      \expandafter\def\csname LT5\endcsname{\color{black}}%
      \expandafter\def\csname LT6\endcsname{\color{black}}%
      \expandafter\def\csname LT7\endcsname{\color{black}}%
      \expandafter\def\csname LT8\endcsname{\color{black}}%
      \fi
      \fi
      \setlength{\unitlength}{0.0500bp}%
      \ifx\gptboxheight\undefined%
      \newlength{\gptboxheight}%
      \newlength{\gptboxwidth}%
      \newsavebox{\gptboxtext}%
      \fi%
      \setlength{\fboxrule}{0.5pt}%
      \setlength{\fboxsep}{1pt}%
      \begin{picture}(4320.00,3600.00)%
        \gplgaddtomacro\gplbacktext{%
          \csname LTb\endcsname%
          \put(773,596){\makebox(0,0)[r]{\strut{}$0.9$}}%
          \put(773,999){\makebox(0,0)[r]{\strut{}$0.95$}}%
          \put(773,1401){\makebox(0,0)[r]{\strut{}$1$}}%
          \put(773,1804){\makebox(0,0)[r]{\strut{}$1.05$}}%
          \put(773,2206){\makebox(0,0)[r]{\strut{}$1.1$}}%
          \put(773,2609){\makebox(0,0)[r]{\strut{}$1.15$}}%
          \put(773,3011){\makebox(0,0)[r]{\strut{}$1.2$}}%
          \put(1207,376){\makebox(0,0){\strut{}$-0.4$}}%
          \put(1810,376){\makebox(0,0){\strut{}$-0.2$}}%
          \put(2414,376){\makebox(0,0){\strut{}$0$}}%
          \put(3018,376){\makebox(0,0){\strut{}$0.2$}}%
          \put(3621,376){\makebox(0,0){\strut{}$0.4$}}%
          \put(2073,72){\makebox(0,0)[l]{\strut{}exclusive}}%
          \put(3671,72){\makebox(0,0)[l]{\strut{}sharing}}%
        }%
        \gplgaddtomacro\gplfronttext{%
          \csname LTb\endcsname%
          \put(176,831){\rotatebox{-270}{\makebox(0,0){\strut{}Coverage}}}%
          \put(988,154){\makebox(0,0){\strut{}$c$}}%
          \put(2414,3269){\makebox(0,0){\strut{}$f(x_1)=1, f(x_2)=0.5$}}%
          \csname LTb\endcsname%
          \put(2936,1209){\makebox(0,0)[r]{\strut{}ESS}}%
          \csname LTb\endcsname%
          \put(2936,989){\makebox(0,0)[r]{\strut{}Welfare Optimum}}%
          \csname LTb\endcsname%
          \put(2936,769){\makebox(0,0)[r]{\strut{}Optimum Coverage}}%
        }%
        \gplbacktext
        \put(0,0){\includegraphics{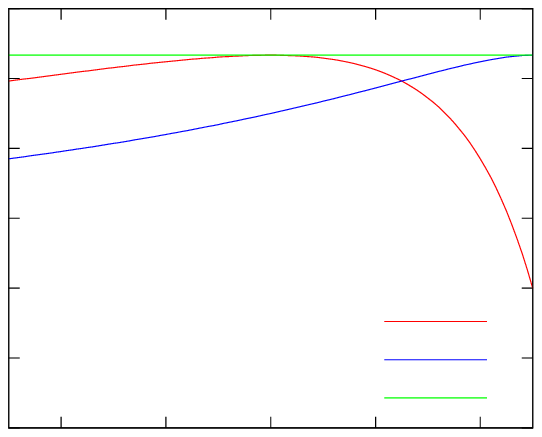}}%
        \gplfronttext
      \end{picture}%
      \endgroup
    \end{subfigure}
    \caption{Coverage as a function of the extent of competition (parametrized by $c$) for the case where two players complete over 
       two sites $x_1$ and $x_2$. On the left (respectively, right) we consider the case where $f(x_1)=1$ and $f(x_2)=0.3$ (respectively, $f(x_2)=0.5$). The $X$-axis represents a range of congestion functions $C_c$, defined as $C_c(1)=1$ and $C_c(2)=c$, where $c\in [-0.5,0.5]$. The case $c=0$ corresponds to the exclusive function and the case $c=0.5$ corresponds to the sharing function. The $Y$-axis represents the coverage. The red line corresponds to the coverage of the ESS, the green corresponds to the optimum coverage, taken over all symmetric strategies, and the blue corresponds to the coverage of the symmetric strategy that maximizes the individual payoffs.}
  \label{fig:coverage}
\end{figure}

\subsection{Related Work} \label{sec:related}
A research theme that is currently quite popular in the Operations Research, and Economic Theory research communities goes by the name ``incentivizing exploration''. Similarly to us, such papers tend to focus on the problem of designing mechanisms to coordinate the activities of independent selfish agents exploring a space of alternatives   \cite{mansour2015bayesian,frazier2014incentivizing,papanastasiou2017crowdsourcing,kremer2014implementing}.
 These papers however typically focus on sequential rather than simultaneous exploration, leading them to model the problem as a multi-armed bandit problem rather than a congestion game. For this reason, there is no real technical connection between the two lines of work, but there is nevertheless a thematic connection.

Closest to our work is the paper by Kleinberg and Oren \cite{kleinberg2011mechanisms} which considers a general model that is motivated by research foundations offering grants to incentivize researchers to work on particular topics. That paper  proposes two
  mechanisms for selecting grant policies that are expected to incentivize researchers to adopt an optimal distribution, while assuming that researchers are restricted to the sharing policy. The first
  mechanism relies on the ability to freely choose the rewards (grants) associated with sites (topics), despite the fact that the significance of a site from the perspective of the group 
  is fixed. In contexts of animals, however, the overall rewards (termed here as $f(x)$) correspond to the amount of food in patch $x$ and therefore cannot be modified. Moreover, implementing their mechanism of choosing rewards requires the knowledge of the number of players $k$, whereas  the specification of the congestion policies studied here does not require this knowledge (see Section \ref{sec:dispersal}). 
  The
  second mechanism proposed in \cite{kleinberg2011mechanisms} requires that the players receive different rewards, set in a particular way by the central entity, despite doing the same actions. This might be seen as unfair and
  unsatisfying. Moreover, implementing this would require high levels of coordination which again seem unlikely in the context of animals.

The settings of selfish routing, job scheduling, and congested games
\cite{monderer1996potential,rosenthal1973congestion} all bear
similarities to the dispersal game, however, the measurement of the
global performances considered here, namely, the coverage, is very
different from the measures studied in the former frameworks, such as
makespan or latency
\cite{makespan1,makespan2,GameTheory,makespan4}. 

Finally, many articles have informally mentioned the IFD as an example of an ESS
\cite{parker1978searching,parker1984evolutionary,pulliam1984living,pulliam1991sources,IFD-review2,morris2003shadows}. However, on a formal level, it is only  recently 
that a rigorous proof was given \cite{cressman2006migration}, and even
this proof was with respect a  limited definition of ESS (essentially,
in our notation, they proved the first item of the characterization with $m_\pi=k-1$, without considering the second item). 
Since that paper, several other works have also addressed this
question rigorously, but, to the best of our knowledge, they all
considered either weaker versions than our ESS definition, or other
contexts, such as continuous populations dynamics
\cite{quijano2007ideal,cosner2012evolutionary}.

\section{The Exclusive Reward Policy}
This section  considers the exclusive reward policy $I_{exc}$.
In Section \ref{sec:sigma} we explicitly calculate the IFD for this
policy, called $\sigma^\star$. Then, in Section \ref{sec:opt-cov}, we prove that $\sigma^{\star}$ yields  optimal coverage. 

\subsection{Algorithm $\sigma^\star$}\label{sec:sigma}
Consider a strategy $p$ that satisfies the IFD conditions for the exclusive reward policy. We know that it exists by Observation \ref{obs} and we calculate it explicitly. In the context of $C_{exc}$, the $value_p$ of a site $x$, as given by Eq. \eqref{p-value}, is:
\[\nu_ p(x)= f(x) (1-p(x))^{k-1}\]
The second IFD condition implies that for any two $x,y\in [1, W]$, we have $f(x) (1-p(x))^{k-1}=f(y) (1-p(y))^{k-1}$, 
or in other words,
\begin{equation}\label{eq-eq}
\frac{1-p(x)}{1-p(y)}=\left(\frac{f(y)}{f(x)}\right)^{1/(k-1)}
\end{equation}
A natural guess for a distribution that satisfies this is the following Pareto distribution:
\begin{align*}  &
 p(x) =
  \begin{cases}
    1 - \frac{\alpha }{f(x)^{\frac{1}{k-1}}} &\quad \text{if } x \leq W\\
    0 &\quad \text{otherwise}
  \end{cases}
\end{align*}
 Since $p$ is a distribution, we must have $\sum_{x=1}^W p(x)=1,$
from which we can extract \[\alpha =\frac{W-1}{\sum_{x \leq W} f(x)^{-\frac{1}{k-1}}}.\] 
To complete the description of our candidate IFD strategy, it remains to calculate $W$.
By the second IFD condition, we get that if $W<M$, then for every $x\leq W$:
\begin{equation}\label{eq-W+1}
f(x) (1-p(x))^{k-1}=\nu_ {p}(x)>\nu_ p(W+1)=f(W+1)
\end{equation}
Rearranging, we get 
$$p(x)<1-\left(\frac{f(W+1)}{f(x)}\right)^{\frac{1}{k-1}}.$$ 
Since $p(x)$ is a distribution whose support is $[W]:=\{1,\ldots, W\}$ we get
$$1=\sum_{x\leq W} p(x)<\sum_{x\leq W}\left(1-\left(\frac{f(W+1)}{f(x)}\right)^{\frac{1}{k-1}}\right).$$
This means that we can define $W$ as the first index that satisfies the equation above, if one exists, or $M$ otherwise. Equivalently, we can define $W$ as the largest index that satisfies
\[ \sum_{x\leq
  W}\left(1-\left(\frac{f(W)}{f(x)}\right)^{\frac{1}{k-1}}\right)\leq
1\] 
Altogether, this  leads to strategy $\sigma^{\star}$, as defined below.

\begin{pseudocode}[H]
  %  \textbf{Algorithm $\sigma^{\star}(x)$}
%\smallskip
 \begin{align*}
 \sigma^{\star}(x) =
  \begin{cases}
    1 - \frac{\alpha}{f(x)^{\frac{1}{k-1}}} &\quad \text{if } x \leq W\\
    0 &\quad \text{otherwise}
  \end{cases} 
  \end{align*}
where $W$ and the normalization factor $\alpha$ are calculated as follows.
\begin{align*}
  &W = \arg\max_{y\leq M} \left\{\sum_{x\leq y} \left(
    1-\left(\frac{f(y)}{f(x)}\right)^{\frac{1}{k-1}}\right) \leq 1\right\}\\
  &\alpha = \frac{W - 1}{\sum_{x \leq W} f(x)^{-\frac{1}{k-1}}}
\end{align*}  
\caption*{Pseudocode of Algorithm $\sigma^{\star}(x)$}
    \label{alg:sigma}
\end{pseudocode}

Interestingly, it turns out that algorithm $\sigma^\star$ is actually
identical to the first round in the algorithm $A^\star$ used in
\cite{SIROCCO} in the context of Bayesian search.

The following claim verifies that the constructed strategy is indeed the IFD. 
\begin{claim}\label{thm-IFD}
$\sigma^{\star}$ satisfies the IFD conditions under the exclusive reward policy $I_{exc}$.  
\end{claim}
\begin{proof}
%\label{app:claim}
The first IFD condition is satisfied by $\sigma^\star$ by
construction. It remains to show that
$\nu_p(1)=\nu_p(2)=\ldots=\nu_p(W)$, and for all $x > W$, we have
$\nu_p(W)>\nu_p(x)$. Under the exclusive reward policy $I_{exc}$, the
$value_p$ of a site $x \leq W$, as given by Eq. \eqref{p-value}, is:
\[
\nu_p(x) = f(x) (1-\sigma^\star(x))^{k-1} = \alpha^{k-1}
\]
which proves the first part. For the second part, we first derive
$f(W+1) < \alpha^{k-1}$ as follows:
\begin{align*}
  &\sum_{x\leq
    W+1}\left(1-\left(\frac{f(W+1)}{f(x)}\right)^{\frac{1}{k-1}}\right)
  > 1\\
 &\sum_{x\leq W}
  \left(1-\left(\frac{f(W+1)}{f(x)}\right)^{\frac{1}{k-1}}\right) > 1
  \\
&\sum_{x\leq W} \left(\frac{f(W+1)}{f(x)}\right)^{\frac{1}{k-1}} <
  W-1\\
&f(W+1)^{\frac{1}{k-1}} < \frac{W-1}{\sum_{x\leq W}
    f(x)^{-\frac{1}{k-1}}}\\
  &f(W+1) < \alpha^{k-1}
\end{align*}
Now it suffices to observe that the $value_p$ of a site $x > W$ is
$\nu_p(x)=f(x)\leq~f(W+1)<\alpha^{k-1}=\nu_p(W)$.
\end{proof}

As claimed in Theorem \ref{thm-ESS}, it turns out that under the
exclusive policy $I_{exc}$, strategy $\sigma^{\star}$ is not only a
Nash Equilibrium, but in fact an ESS.

\subsection{$\sigma^{\star}$ has Optimal  Coverage}\label{sec:opt-cov}
In this subsection we prove Theorem \ref{thm-optimum}, that is, we prove that  for any strategy $\sigma$, ${\Cover}(\sigma^{\star})\geq {\Cover}(\sigma)$, with equality if and only if $\sigma=\sigma^{\star}$.
Before dwelling into the proof we note that algorithm
$A^\star$ was shown in \cite{SIROCCO} to be optimal for the setting in which $k$ independent searchers are searching for a treasure placed in one of $M$ boxes, according to a Bayesian distribution. The proof therein is long and relies on deep techniques
that were developed in \cite{STOC}. Since $\sigma^*$ agrees with $A^\star$ on the first round of $A^\star$ and since there are some similarities between the search and the coverage objectives, it might be possible to translate
the proof in \cite{SIROCCO,STOC} to show the
optimality of $\sigma^\star$ in terms of coverage. However, we managed to find a simple and concise proof of this
fact, that can also be used to show the uniqueness of $\sigma^\star$.
We provide it below. 
 
Note that maximizing ${\Cover}(p)$ is equivalent to  minimizing
${\cal{T}}(p)= \sum_x f(x)(1-p(x))^k$. We next show that $\sigma^{\star}$ minimizes the latter expression.
 
 First observe that since the number of resources is finite, it follows by a compactness argument  that 
the infimum of $\left\{{\cal{T}}(\sigma)\mid {\sigma \mbox{~is a distribution}}\right\}$ is attained by some distribution $\sigma_{\min}$.   
Assume by contradiction that $\sigma_{\min}\neq \sigma^{\star}$. Since both $\sigma_{\min}$ and $\sigma^{\star}$ are distributions, we must have some $x_0\leq W$ such that $\sigma^{\star}(x_0)>\sigma_{\min}(x_0)$ and some $x_1$ such that $\sigma^{\star}(x_1)<\sigma_{\min}(x_1)$. 
For sufficiently small $\epsilon>0$, let us define $\sigma'(x)$ which equals $\sigma_{\min}$ everywhere except that $\sigma'(x_1)=\sigma_{\min}(x_1)-\epsilon$ and $\sigma'(x_0)=\sigma_{\min}(x_0)+\epsilon$. In other words, we create $\sigma'$ from $\sigma_{\min}$ by shifting a small mass of probability from $x_1$ to $x_0$. 

Aiming to contradict the minimality of $\sigma_{\min}$, our  goal is to show that ${\cal{T}}(\sigma')<{\cal{T}}(\sigma_{\min})$. In other words, we aim to show that 
the following expression, which equals to ${\cal{T}}(\sigma_{\min})-{\cal{T}}(\sigma')$, is positive. 
\begin{equation}
  f(x_0) \left((1-\sigma_{\min}(x_0))^k-
  (1-\sigma_{\min}(x_0)-\epsilon)^k \right)
  +f(x_1) \left((1-\sigma_{\min}(x_1))^k-
  (1-\sigma_{\min}(x_1)+\epsilon)^k \right)
\end{equation}
By a Taylor expansion, for a small $\epsilon>0$, we obtain:
\begin{equation}
  {\cal{T}}(\sigma_{\min})-{\cal{T}}(\sigma')\approx \epsilon k \left(
  f(x_0)  (1-\sigma_{\min}(x_0))^{k-1}\right.
  \left. -  f(x_1)  (1-\sigma_{\min}(x_1))^{k-1}   \right) 
\end{equation}
By Eq.  \eqref{eq-W+1} and the definitions of $x_0$ and $x_1$, we know that 
\begin{equation}
  f(x_0) (1-\sigma_{\min}(x_0))^{k-1}>f(x_0)
  (1-\sigma^{\star}(x_0))^{k-1}
  \geq f(x_1) (1-\sigma^{\star}(x_1))^{k-1} >f(x_1)
  (1-\sigma_{\min}(x_1))^{k-1}
\end{equation}
 Hence, for sufficiently small $\epsilon>0$, we get that ${\cal{T}}(\sigma_{\min})>{\cal{T}}(\sigma')$, contradicting the minimality of $\sigma_{\min}$.  This establishes the fact that  $\sigma^{\star}=\sigma_{\min}$, implying that $\sigma^{\star}$ maximizes the coverage ${\Cover}(p)$. 
 
 Finally,  the fact that  $\sigma^{\star}=\sigma_{\min}$ also implies that the distribution that minimizes ${\cal{T}}(p)$ (and hence maximizes ${\Cover}(p)$) is unique. Hence,
 ${\Cover}(\sigma^{\star})= {\Cover}(\sigma)$ can occur only if $\sigma=\sigma^{\star}$. This concludes the proof of Theorem \ref{thm-optimum}. \qed
%\end{proof}

\section{The Ideal Free Distribution is an ESS}\label{app:ESS}
The goal of this section is to prove Theorem \ref{thm-ESS}, i.e, we aim to show that under $I_{exc}$, strategy 
 $\sigma^{\star}$ is an ESS.  

Let $\sigma$ be a strategy, distinct from $\sigma^\star$.
First note that, when playing against strategy $\sigma^\star$, because
it is the IFD, the expression of
$\mathcal{E}_k(\sigma;\sigma^{\star(k-1)})$ becomes :
\begin{equation}
  \mathcal{E}_k(\sigma;\sigma^{\star(k-1)}) = \sum_{x = 1}^M \sigma(x) \nu_{\sigma^\star}(x)
  = \sum_{x = 1}^W \sigma(x) \nu_{\sigma^\star}(W) + \sum_{x = W +1}^M \sigma(x) f(x)
\end{equation}
Since for every site $x > W$, we have $f(x) < \nu_{\sigma^\star}(W)$,
it is clear that every strategy $\sigma$ whose support is not a subset
of $[1,W]$ has
\[{\cal E}_k(\sigma;\sigma^{\star(k-1)}) < {\cal
  E}_k(\sigma^{\star};\sigma^{\star(k-1)})\] and therefore satisfies
the ESS characterization conditions with $m_\sigma = 0$.

Next, let us assume that the support of $\sigma$ is included in
$[1,W]$, and show that the ESS characterization conditions are satisfied for
$m_\sigma = 1$, i.e., we aim to show that
\begin{align}
  &\forall \ell, 1\leq \ell \leq k-2,\mbox{ we have }\notag\\
  &{\cal{E}}_k(\sigma^{\star};\sigma^\ell,\sigma^{\star(k-\ell-1)})>{\cal{E}}_k(\sigma;\sigma^\ell,\sigma^{\star(k-\ell-1)}) \label{eq:first},\\
  &\text{ and } {\cal{E}}_k(\sigma^{\star};\sigma^{\star(k-1)})={\cal{E}}_k(\sigma;\sigma^{\star(k-1)}).\label{eq:second}
\end{align}
In fact, to show that the ESS characterization conditions are satisfied for
$m_\sigma = 1$, we need only the case $\ell=1$ in Inequality \eqref{eq:first}, but we prove here a stronger stability criteria, showing that this inequality holds for every $1\leq \ell \leq k-2$. 

Eq.  \eqref{eq:second} follows quite simply from the observation above,
coupled with the assumption that the support of $\sigma$ is in
$[1,W]$, which gives both ${\cal{E}}_k(\sigma^{\star};\sigma^{\star(k-1)}) = \nu_{\sigma^\star}$ and ${\cal{E}}_k(\sigma;\sigma^{\star(k-1)}) = \nu_{\sigma^\star}$.
     
Eq. \eqref{eq:first} requires more elaboration. We first show that for every $1\leq \ell\leq k-2$: \[{\cal{E}}_k(\sigma^{\star};\sigma^\ell,\sigma^{\star(k-\ell-1)})
 \geq  {\cal{E}}_k(\sigma;\sigma^\ell,\sigma^{\star(k-\ell-1)})\]  and then show that equality holds iff $\sigma=\sigma^\star$.
We start with the following claim.\\

\begin{claim} \label{clm:expanding}
Expanding each side in Eq. \eqref{eq:first}, we obtain:
\begin{align}
  {\cal{E}}_k(\sigma^{\star};\sigma^\ell,\sigma^{\star(k-\ell-1)}) &=  
  {\alpha^{k-\ell-1}}\left(\sum_{x \leq W} f(x)^{\frac{\ell}{k-1}}
  (1-\sigma(x))^{\ell}-
  {\alpha} \sum_{x \leq
    W}f(x)^{\frac{\ell-1}{k-1}}(1-\sigma(x))^{\ell}\right),\\
  {\cal{E}}_k(\sigma;\sigma^\ell,\sigma^{\star(k-\ell-1)}) &=
  {\alpha^{k-\ell-1}} \left(\sum_{x \leq W} f(x)^{\frac{\ell}{k-1}}
  (1-\sigma(x))^{\ell} -
  \sum_{x \leq W} f(x)^{\frac{\ell}{k-1}} (1-\sigma(x))^{\ell+1}\right)
\end{align}
\end{claim}

\noindent{\em Proof of Claim \ref{clm:expanding}.}
We start by expanding the left hand side of
Eq. \eqref{eq:first}.
\begin{align*}
  {\cal{E}}_k(\sigma^{\star};\sigma^\ell,\sigma^{\star(k-\ell-1)}) &=
  \sum_x f(x) \sigma^{\star}(x) (1-\sigma^{\star}(x))^{k-\ell-1} (1-\sigma(x))^{\ell}\\
  &= \sum_{x \leq W} (f(x) - \alpha f(x)^{\frac{k-2}{k-1}}) {\alpha^{k-\ell-1}}{f(x)^{\frac{\ell+1-k}{k-1}}}(1-\sigma(x))^{\ell}\\
  &= {\alpha^{k-\ell-1}} \sum_{x \leq W} (f(x)^{\frac{\ell}{k-1}} - \alpha f(x)^{\frac{\ell-1}{k-1}})(1-\sigma(x))^{\ell}\\
  &= {\alpha^{k-\ell-1}}\left(\sum_{x \leq W} f(x)^{\frac{\ell}{k-1}}
  (1-\sigma(x))^{\ell} - \alpha  \sum_{x \leq W} f(x)^{\frac{\ell-1}{k-1}}(1-\sigma(x))^{\ell}\right)
\end{align*}

This establishes the first part of the claim.\\

Next, we expand the right hand side of Eq. \eqref{eq:first}.
\begin{align*}
  {\cal{E}}_k(\sigma;\sigma^\ell,\sigma^{\star(k-\ell-1)}) &=
  \sum_x f(x) \sigma(x) (1-\sigma^{\star}(x))^{k-\ell-1}(1-\sigma(x))^{\ell}\\
  &= \sum_{x \leq W} f(x) \sigma(x) (1-\sigma^{\star}(x))^{k-\ell-1}(1-\sigma(x))^{\ell}\\
  &= \sum_{x \leq W} f(x) \sigma(x)  {\alpha^{k-\ell-1}}{f(x)^{\frac{\ell+1-k}{k-1}}}(1-\sigma(x))^{\ell}\\
  &= {\alpha^{k-\ell-1}}\sum_{x \leq W}  \sigma(x)  {f(x)^{\frac{\ell}{k-1}}}(1-\sigma(x))^{\ell}\\
  &= {\alpha^{k-\ell-1}} \left(\sum_{x \leq W}
    f(x)^{\frac{\ell}{k-1}} (1-\sigma(x))^{\ell} -
  \sum_{x \leq W} f(x)^{\frac{\ell}{k-1}} (1-\sigma(x))^{\ell+1}\right)
\end{align*}
 Concluding the proof of Claim \ref{clm:expanding}. 
\bigskip

By Claim \ref{clm:expanding}, it follows that the desired inequality is equivalent to:
\begin{equation*}
  \alpha \sum_{x \leq W}
  f(x)^{\frac{\ell-1}{k-1}}(1-\sigma(x))^{\ell} \leq \sum_{x \leq
    W} f(x)^\frac{\ell}{k-1} (1-\sigma(x))^{\ell+1}
\end{equation*}
By definition of $\alpha$, we rearrange this inequality into
\begin{equation*}
  (W - 1) \sum_{x \leq W}
  f(x)^{\frac{\ell-1}{k-1}}(1-\sigma(x))^{\ell} \leq \sum_{x \leq W}
  \frac{1}{f(x)^{\frac{1}{k-1}}} \sum_{x \leq W}
  f(x)^{\frac{\ell}{k-1}} (1-\sigma(x))^{\ell+1}
\end{equation*}
Decomposing
$f(x)^{\frac{\ell-1}{k-1}}=\frac{f(x)^{\frac{\ell}{\ell+1}\cdot
    \frac{\ell}{k-1}}}{f(x)^{\frac{1}{(\ell+1)(k-1)}}}$, and
applying H\"older's Inequality, we obtain:
\begin{equation*}
  \sum_{x \leq W} f(x)^{\frac{\ell-1}{k-1}}(1-\sigma(x))^{\ell} \leq
  \left( \sum_{x \leq W} \frac{1}{f(x)^{\frac{1}{k-1}}}
  \right)^{\frac{1}{\ell+1}} \left(\sum_{x \leq W}
  f(x)^{\frac{\ell}{k-1}}
  (1-\sigma(x))^{\ell+1}\right)^{\frac{\ell}{\ell+1}}
\end{equation*}
This implies that to prove that Eq.  \eqref{eq:first} holds,  it is sufficient to establish the following inequality:
\begin{equation*}
  (W-1) \left( \sum_{x \leq W} \frac{1}{f(x)^{\frac{1}{k-1}}} \right)^{\frac{1}{\ell+1}}  \leq
  \sum_{x \leq W} \frac{1}{f(x)^{\frac{1}{k-1}}} \left(\sum_{x \leq W} f(x)^{\frac{\ell}{k-1}}
  (1-\sigma(x))^{\ell+1}\right)^{\frac{1}{\ell+1}}
\end{equation*}
Using the fact that $\sigma$ is a distribution, and hence $ \sum_{x \leq W} (1-\sigma(x))=W-1$, it suffices to prove that
\begin{equation*}
  \sum_{x \leq W} (1-\sigma(x)) \left( \sum_{x \leq W} \frac{1}{f(x)^{\frac{1}{k-1}}} \right)^{\frac{1}{\ell+1}} \leq
  \sum_{x \leq W} \frac{1}{f(x)^{\frac{1}{k-1}}} \left(\sum_{x \leq W} f(x)^{\frac{\ell}{k-1}}
  (1-\sigma(x))^{\ell+1}\right)^{\frac{1}{\ell+1}}
\end{equation*}
Rearranging, this is equivalent to proving that 
\begin{equation}\label{holder}
  \sum_{x \leq W} (1-\sigma(x))\leq  \left( \sum_{x \leq W} \frac{1}{f(x)^{\frac{1}{k-1}}} \right)^{\frac{\ell}{\ell+1}}
  \left(\sum_{x \leq W} f(x)^{\frac{\ell}{k-1}}
  (1-\sigma(x))^{\ell+1}\right)^{\frac{1}{\ell+1}}
\end{equation}
Which finally follows by applying  H\"older's Inequality.

The fact that the inequality is strict for $\ell\geq 1$ follows from the following.
\begin{claim}\label{claimlinear}
For $\ell\geq 1$, $ ~{\cal{E}}_k(\sigma^{\star};\sigma^\ell,\sigma^{\star(k-\ell-1)}) =
{\cal{E}}_k(\sigma;\sigma^\ell,\sigma^{\star(k-\ell-1)}) ~\implies~\sigma = \sigma^{\star}$.
\end{claim}

\noindent{\em Proof of Claim \ref{claimlinear}.}\label{app:claimlinear}
The proof of the claim is based on the observation that the former
equality implies, in particular, an equality in H\"older's inequality
in Eq. \eqref{holder}. For $\ell\geq 1$, this can be the case only if the two sides are
linearly dependent, i.e., for each $x\leq W$, we have \[f(x)^{\frac{\ell}{k-1}} (1-\sigma(x))^{\ell+1} = \frac{\beta^{\ell+1}}{f(x)^{\frac{1}{k-1}}}\] where $\beta$ is some fixed constant. 
Together with the fact that the support of $\sigma$ is included in  $[W]$, the strategy $\sigma$ must
therefore be of the form
  \begin{equation*}
    \sigma(x) =
    \begin{cases}
      1-\frac{\beta}{f(x)^{\frac{1}{k-1}}} &\text{if } x \leq W\\
      0 &\text{otherwise}
    \end{cases}
  \end{equation*}
  Finally, we know that \[\sum_{x\leq W} \sigma(x) = \sum_{x\leq
    W} \left(1-\frac{\beta}{f(x)^{\frac{1}{k-1}}}\right) = 1\] Solving this for $\beta$ gives
  \[\beta = \frac{W - 1}{\sum_{x \leq W} f(x)^{-\frac{1}{k-1}}} =
  \alpha\] and therefore $\sigma = \sigma^{\star}$.  This completes the proof of Claim \ref{claimlinear}, and thus established the proof for Theorem~\ref{thm-ESS}. \qed

\section{The Criticality of the Exclusive Collision Cost Function}\label{sec:criticality}
The goal in this section is to prove Theorem \ref{thm-critical}. That is,
 fixing a congestion function $C(\cdot)\neq C_{exc}(\cdot)$, we aim to show that  $\SPoA(C)>1$.

  First recall that by definition of congestion functions, $C(\cdot)$ is non-increasing. Hence Observation~\ref{obs} applies, and the existence of the IFD is guaranteed for every value function $f$. Our plan is to show that there exists a value function $f$, for which this IFD is different than $\sigma^{\star}$.  
 Once this is established, the fact that $\SPoA(C)>1$ follows from Theorem \ref{thm-optimum}.

Assume, towards contradiction, that for every $f$, the IFD under $C(\cdot)$ is precisely $\sigma^{\star}$.
Let $M\gg k$, and let $f$ be a strictly decreasing function, that, nevertheless, decreases very slowly, such that, e.g., for every $x\leq y$ in $[1,M]$, we have 
%\begin{equation}\label{eq:decrease}
$\frac{f(y)}{f(x)}\geq \frac{f(M)}{f(1)}> \left(1-\frac{1}{2k}\right)^{k-1}$. 
%\end{equation}
Together with the definition of $W$ (see the pseudocode  of Algorithm $\sigma^{\star}(x)$),  
 we obtain that $W\geq 2k$. 
 
  Because $\sigma^{\star}$ is the IFD, we have that $\forall x
  \leq W$, $\nu_p(x) = \nu_p(W)$. Let us develop this expression:
  \begin{align*}
    \nu_p(W) &= f(x) \sum_{\ell = 1}^{k} C(\ell) P(x, \ell - 1)\\
    &= f(x) \sum_{\ell = 1}^{k} C(\ell) \binom{k-1}{\ell-1}
    {\sigma^{\star}}(x)^{\ell-1} (1-\sigma^{\star}(x))^{k-\ell}\\
    &= f(x) \sum_{\ell = 1}^{k} C(\ell) \binom{k-1}{\ell-1}
    \left(1-\alpha f(x)^{-\frac{1}{k-1}}\right)^{\ell-1} \alpha^{k -
      \ell} f(x)^{-\frac{k-\ell}{k-1}}\\
    &=\sum_{\ell = 1}^{k} C(\ell) \binom{k-1}{\ell-1} \left(1-\alpha
    f(x)^{-\frac{1}{k-1}}\right)^{\ell-1} \alpha^{k - \ell}
    f(x)^{\frac{\ell-1}{k-1}}
  \end{align*}
The last line above is  a polynomial equation in $f(x)^{\frac{1}{k-1}}$ of degree at most $k-1$.
 Note that this polynomial is not a constant. Indeed, since $C\neq C_{exc}$ and $C(1)=C_{exc}(1)=1$, there must exist $\ell\geq 2$, such that $C(\ell)\neq 0$. Hence, the $\ell-1$'st coefficient of the polynomial is non-zero, implying that 
    this polynomial is not a constant. It follows that there are at most $k-1$ of values for
  $f(x)$ that would satisfy this equation (recall that $f$ is strictly decreasing). This yields a contradiction since the equation is supposed to be true for
  every site $x \leq W$, and $W\geq 2k$.  This concludes the proof of Theorem \ref{thm-critical}. \qed
 % \end{proof}

\section{Discussion}
\subsection{Conclusions and Future Directions}\label{sec:conclusions}

This paper focuses on the mechanism design challenges in the basic dispersal game inspired by animal behavior. Our focus is on congestion policies, where there is no control over the utility associated with a site, and only the price associated with congestion can vary. Indeed, for natural scenarios in particular, it makes sense that evolution would impact the competition levels but the values of sites would be determined by the environment. 
The definition of congestion policies neglects, however, a factor that might play a significant role in several scenarios, namely, the cost incurred when in visiting a site $x$ (e.g., the energetic cost consumed while traveling to $x$). Studying the more general model that takes into account this extra cost is left for future work.

The coverage measure can find relevance in contexts of animal foraging. This is because effective group coverage can indirectly contribute to the fitness of individuals, especially when the group is in competition with other groups over the same set of resources \cite{Social-foraging}. Indeed, in this case, the consumption of many resources by conspecifics reduces the fitness of individuals in the competing group and vice versa. 
 It is therefore of interest to develop a more comprehensive game-theoretic  perspective that would integrate into the individual fitness both the competition inside the group, and the competition between groups \cite{Eldar}. A plausible insight that can be learned from our results is that aggressive behavior between conspecifics, which appears to be wasteful or even harmful from the group's perspective, can in fact be effective on the group level, as it allows for better coverage. 
For further discussion on the interpretation of our results in the context of animals see Section \ref{sec:discussion}. 

Our definition of coverage assumes that a single player in a site $x$ suffices to consume the full utility $f(x)$. In natural settings, e.g., in animal foraging scenarios, this assumption might be too strong, and it is of interest to study relaxations of it. One possibility is to assume that there is a maximum capacity of consumption per individual. 
In some other cases, a more complex definition of coverage might be appropriate. We leave the study of such generalizations for future work.

Although we formally discuss only the one-shot scenario, we wish to stress that the ESS, by definition, is stable when played repeatedly in an infinite population. We did not, however, investigate other forms of repetition, which are left for future work.

Finally, game theory in computer science typically studies the notion of Nash Equilibrium as the main concept to capture stability. This notion was originally designed for games with small number of players. When considering games with thousands or even millions of players, it becomes less plausible to assume that all players are rational, and one may seek for stronger notions of stability. ESS, which was developed in the evolutionary game theoretical community, appears to be a good candidate to capture stability in large populations. 
 We believe that this notion should be applied more in game theory works that correspond to large distributed settings

\subsection{Informal Discussion on Animal Dispersal}\label{sec:discussion}
From the perspective of the group, consuming a large amount of food by all members together can indirectly increase the fitness of the group members and hence become significant for their survival.  For example, when multiple species compete over the same patched resource, a thorough consumption of patches by one species implies less food remaining for the other.
 In this respect, it would be interesting to experiment on the interplay between two similar species that compete over the same set of resources, but differ in their level of aggressiveness toward conspecifics \cite{noga}. In cases where there is no direct contact between the species (e.g., when they feed at different times of the day), it may appear that the more aggressive species would be inferior as it induces unnecessary waste of energy and risks of injury. However, our results suggest that perhaps it is the aggressive species that would turn out to be the superior one. Indeed, its aggressive behavior incurs higher collisions costs, which may drive individuals to better cover the resources, on the expense of the more peaceful species.   
 
During foraging, both selfish individuals and collaborating groups (e.g., ants or bees) share the basic
 challenge of balancing the  need to exploit the more promising resources, while trying to avoid overlaps (or collisions) with conspecifics \cite{Social-foraging,ANTS}. 
This tradeoff, however, is manifested in the two cases in a manner that appears to be very different: While collaborative groups aim to optimize the tradeoff by setting the parameters in a somewhat ``centralized'' manner, the competing  dynamics is governed by a selection process  yielding an evolutionary stable strategy which need not be optimal for  individuals. 
As suggested here, despite these differences, when colliding individuals are punished severely, the two perspectives yield similar strategies. 

Finally, the alignment in behavior between collaborating and competing individuals, as reported here, may  suggest a modest contribution to the theory of evolution of eusocial species. 
The currently dominated theory proposes that  eusocial species have evolved from self motivated individuals 
  due to kin selection \cite{Hamilton}. % or by multilevel selection \cite{}.
An issue that is discussed less is that a shift in the ``motivation focus'' can potentially be accompanied by a  shift in function and hence, perhaps, also in structure \cite{MCNAMARA}. Indeed, in principle, while evolving from self-interested to altruistic, an individual which had evolved in a way that was tailored to its own needs, may now need to adapt to the collective needs. As these needs are on different scales, they may not necessarily be akin \cite{JEB}, and hence a long evolutionary process may  be required in order to bridge this gap.   
In this respect,  that seemingly cooperative behavior can arise in competing systems, as demonstrated here, suggests that the behavioral transition, at least in the context of  dispersal or forage, may have been smoother than what one might expect.

\bigskip
\paragraph{Acknowledgements.} The authors are thankful for Yoav Rodeh, Yossi Yovel, Yuval Emek, Lucas Boczkowski, Emanuele Natale, and Ofer Feinerman for useful discussions concerning both the model and its applications to biological contexts. 

%\clearpage
%\clearpage
\bibliographystyle{plain}
\bibliography{bib}
%%%%%%%%%%%%%%%%%%%%%%%%%%%%%%%%%%%%%%%%%%%%%%%

%\newpage
%\bigskip
\clearpage
\centerline{\Large\bf APPENDIX}
\bigskip
\appendix
\bigskip

\section{Proof Sketch for Observation \ref{obs}}\label{app:obs}

%\begin{proof}
The existence of the IFD follows by the fact that $\nu_p(x)$ is non-increasing with $p(x)$.
  By definition, the IFD is a symmetric Nash Equilibrium of the
  game.  We now prove its
  uniqueness. Imagine there are two different symmetric Nash Equilibria,
  $\pi$ and $\pi'$, respectively associated with their values
  $\nu_\pi$ and $\nu_{\pi'}$. Note that for any fixed $x$, the $value_p$
  $\nu_p(x)$ is a strictly decreasing function of $p(x)$. Therefore
  \[\pi = \pi' \iff \nu_{\pi} = \nu_{\pi'}\] Assume by contradiction
  that $\nu_{\pi} > \nu_{\pi'}$. This implies that for any $x \in
  [\pi],\ \pi(x) < \pi'(x)$, where $[\pi]$ denotes the support of
  $\pi$, i.e the set of sites explored with positive
  probability. Summing over all sites in $[\pi]$, we get \[\sum_{x \in
    [\pi]} \pi(x) < \sum_{x \in [\pi]} \pi'(x)\] Since the left-hand
  size is equal to 1, we get a contradiction. This proves the uniqueness
  of the symmetric Nash Equilibrium, and of the IFD. \qed
%\end{proof}
\bigskip

\end{document}